\documentclass[12pt]{article}
\usepackage{amsmath}
\usepackage{amsthm}
\usepackage{graphicx}
\usepackage{enumerate}
\usepackage[colorinlistoftodos]{todonotes}
\usepackage[a4paper, margin=1in]{geometry}
\usepackage{amssymb}
\usepackage{multirow}
\usepackage{bm}
\usepackage{natbib}
\usepackage{times}
\usepackage{caption}
\usepackage{subcaption}
\usepackage[plain,noend]{algorithm2e}
\usepackage{float}
\newtheorem{thm}{Theorem}

\begin{document}
\title{Adjusted Empirical Likelihood for Long-memory Time Series Models}
\author{Ramadha D. Piyadi Gamage, Wei Ning \footnote{Corresponding author. Email:
wning@bgsu.edu}\hskip .1in and Arjun K. Gupta\\ Department of
Mathematics and Statistics\\Bowling Green State University, Bowling
Green, OH 43403, USA}

\date{}
\maketitle

\begin{abstract}
\noindent Empirical likelihood method has been applied to short-memory time series models by Monti (1997) through the Whittle's estimation method. Yau (2012) extended this idea to long-memory time series models. Asymptotic distributions of the empirical likelihood ratio statistic for short and long-memory time series have been derived to construct confidence regions for the corresponding model parameters. However, computing profile empirical likelihood function involving constrained maximization does not always have a solution which leads to several drawbacks. In this paper, we propose an adjusted empirical likelihood procedure to modify the one proposed by Yau (2012) for autoregressive fractionally integrated moving average (ARFIMA) model. It guarantees the existence of a solution to the required maximization problem as well as maintains same asymptotic properties obtained by Yau (2012). Simulations have been carried out to illustrate that the adjusted empirical likelihood method for different long-time series models provides better confidence regions and coverage probabilities than the unadjusted ones, especially for small sample sizes. \\

\noindent \textbf{Keywords}: Adjusted empirical likelihood; ARFIMA models; Bartlett correction; Confidence regions; Coverage probability; Whittle's likelihood.
\end{abstract}

\section{Introduction}
\noindent
Owen (1988, 1990, 1991) introduced empirical likelihood (EL) method which is the data-driven method combining the advantages of parametric and nonparametric methods. The most appealing property of the EL method is that the associated empirical likelihood ratio statistics asymptotically follows standard chi-square distribution, which is same as the one used in parametric analysis. Since then, it has been widely used to make statistical inference of parameters and construct confidence regions. See Owen (2001) for more details. However, when the data is dependent, it becomes difficult to apply the empirical likelihood method as it is originally designed for independent observations. Using EL method to address dependent data problems has been studied by many researchers. Mykland (1995) established the connection between the dual likelihood and the empirical likelihood through the martingale estimating equations and applied it to time series model. Monti (1997) developed the idea of extending the EL method to short-memory stationary time series by using the Whittle's (1953) method to obtain an M-estimator of the periodogram ordinates of time series models which are asymptotically independent. However, his method can not be applied directly to long-time memory time series model. Kitamura (1995) developed the blockwise empirical likelihood method for time series models. For long-memory or long-range dependence time series data, Hurvich and Beltrao (1993) showed that the normalized periodogram ordinates obtained from a Gaussian process are asymptotically neither independent identically distributed nor exponentially distributed. Nordman and Lahiri (2006) developed frequency domain empirical likelihood based on the spectral  distribution through the fourier transformation to study short and long range dependence. Yau (2012) extended Monti's idea to autoregressive fractionally integrated moving average (ARFIMA) model by showing that the dependence in periodogram only applies to a small portion of the periodogram ordinates with fourier frequencies tending to zero. However, the profile empirical likelihood function computation which involves constrained maximization requires the convex hull of the estimating equation to have zero vector as an interior point. When the solution does not exist, Owen (2001) suggested assigning $-\infty$ to the log-EL statistic. Chen et al. (2008) pointed out the drawbacks in doing so and proposed an adjusted empirical likelihood (AEL) method by adding a pseudo term which always guarantees the existence of a solution. They further showed that the asymptotic results of the AEL are similar to that of the EL. Moreover, it achieves improved coverage probabilities without using Bartlett-corrections. Based on their work, Piyadi Gamage et al. (2016) modified Monti's work by proposing an adjusted empirical likelihood for short-memory time series models. 

In this paper, we extend Yau's EL method for ARFIMA model by proposing an adjusted empirical likelihood method. The rest of the article is organized as follows. In Section 2, the EL for ARFIMA models is discussed. The AEL for stationary ARFIMA model is derived and the asymptotic distribution of the AEL statistic is established in Section 3. Simulations are carried out in Section 4 to compare the confidence regions of the proposed AEL method to the EL method for ARFIMA model for different distributions for the white noise term. In addition, coverage probabilities are calculated to illustrate the effectiveness of AEL method as compared to EL method with and without Bartlett-correction for different values of the parameters, different sample sizes and different distributions for the white noise term. Section 5 provides some discussion and proofs of results are given in the Appendix.

\section{Empirical Likelihood for ARFIMA Models}
A stationary ARFIMA(p,d,q) process~$Z_t$~is given by 
\begin{equation}
\Phi(B) (1-B)^d Z_t = \Theta(B)a_t  \tag{2.1}\label{eq:2.1}
\end{equation}
for some $-\frac{1}{2} < d < \frac{1}{2}$ where $B$ is the backward shift operator ($BZ_t=Z_{t-1}$), with $\Theta(B) = 1+\theta_1B + \theta_2B^2 + ... + \theta_qB^q$ 
and $\Phi(B)=1-\phi_1B - \phi_2B^2 - ... - \phi_p B^p.$ The absolute values of the roots of these two polynomials are all greater than 1 to guarantee the stationarity and invertibility of the model. We also assume that~$\Theta(B)$~and~$\Phi(B)$~have no common factors to avoid the redundancy of the parameters. 
We only consider the values of~$d$~between 0 and 0.5 since it is the most interesting long-memory scenario (Beran, 1994)  and only under this condition the dependence structure of periodogram ordinates has been established (Yau, 2012). The parameter $\beta =$ $(\phi_1,...,\phi_p,\theta_1,...,\theta_q,d, \sigma^2)^\prime \in \mathbb{B}$ is estimated by Whittle's method (Whittle, 1953) based on the periodogram where $\mathbb{B}$ is a compact subset of the $k$-dimensional Euclidean space $(k=p+q+2)$ .

Let $z_1, z_2, ..., z_T$ be $T$ observations from the process in \eqref{eq:2.1}. An approximate log-likelihood function is given by Whittle (1953),
\begin{equation}
\ln \{L(\beta)\} = -\sum_{j=1}^{n}\ln\{g_j(\beta)\} -\sum_{j=1}^{n}\frac{I(\omega_j)}{g_j(\beta)}, \tag{2.2}\label{eq:2.2}
\end{equation}
where $g_j(\beta)$ is the spectral density and
$$I(\omega_j) = \frac{1}{2\pi T} \bigg[ {\bigg\{\sum_{t=1}^{T}(z_t-\bar{z})\sin(\omega_j t)\bigg\} }^2 + {\bigg\{\sum_{t=1}^{T}(z_t-\bar{z})\cos(\omega_j t)\bigg\} }^2 \bigg]$$
is the periodogram ordinate evaluated at fourier frequency $\omega_j=2\pi j/T,j=1,2,...,n=\frac{T-1}{2}.$ The Whittle's estimator, $\beta_n$, maximizes \eqref{eq:2.2} over $\mathbb{B}$. Therefore, it is the solution of~$\psi(\beta)=\sum_{i=1}^n \psi_j\{I(\omega_j, \beta)\}=0,$~where 
\begin{equation}
\psi_j\{I(\omega_j),\beta\}=\bigg\{\frac{I(\omega_j)}{g_j(\beta)}-1\bigg\} \frac{\partial\ln\{g_j(\beta)\}}{\partial \beta}, \tag{2.3}\label{eq:2.3}
\end{equation}
Monti (1997) showed that this estimator has the interpretation of an M-estimator from asymptotically independent periodogram ordinates and applied the empirical likelihood to short-memory time series models. As pointed out by Yau (2012), the EL method used by Monti (1997) cannot be directly applied to long-memory scenario due to the dependence structure of the periodogram. As shown in Lemma 1 of Yau (2012), the periodogram ordinates are asymptotically independent which ensures that the empirical likelihood ratio statistic to be chi-squared distributed and uses the $\psi_j$'s in equation \eqref{eq:2.3} to construct the empirical likelihood.

Following Monti's (1997) argument, Yau (2012) extended the empirical likelihood ratio statistic to the ARFIMA(p,d,q) model defined by
\begin{equation}
{W}(\beta) = -2 \ln R(\beta) = 2 \, \sum_{j=1}^{n} \ln \big[ 1 + {\xi}(\beta)^\prime\psi_j(I(\omega_j),\beta)\big], \tag{2.4}\label{eq:2.4}
\end{equation}
where ${\xi}(\beta)$ is the Lagrangian multiplier satisfying $\sum_{j=1}^{n}\frac{\psi_j(I(\omega_j),\beta)}{1 + {\xi}(\beta)^\prime\psi_j(I(\omega_j),\beta)}=0$ and ${R}(\beta)$ is the empirical likelihood ratio for time series models defined by
$$R(\beta) = \left. {\sup \prod_{j=1}^{n} p_j} \middle/ {\prod_{j=1}^{n} \frac{1}{n}} \right. = \sup \prod_{j=1}^{n} np_j,$$
subject to the constraints:(i) $\sum_{j=1}^{n}\psi_j(I(\omega_j),\beta) p_j=0$, (ii) $\sum_{j=1}^{n} p_j=1$, and (iii) $p_j \geq 0 , j=1,2,...,n$ with $ p_j = [n\{1 + \xi(\beta)\psi_j(I(\omega_j),\beta)\}]^{-1}$. By extending Monti's result to ARFIMA models, Yau (2012) showed that ${W}(\beta)$ also has an asymptotic chi-squared distribution with $k$ degrees of freedom where $k=p+q+2$.

\section{Adjusted Empirical Likelihood for ARFIMA Models}
The definition of ${W}(\beta)$ in (2.4) depends on obtaining positive $p_js$ such that
\begin{align*} 
\sum_{j=1}^{n}\psi_j(I(\omega_j),\beta) p_j=0,
\end{align*}
for each $\beta$. Under some moment conditions on $\psi_j(x_j,\beta)$ (Owen 2001), the solution exists if the convex hull $\{\psi_j(x_j,\beta), j=1,2,...,n\}$ contains 0 as its interior point with probability 1 as $n \rightarrow \infty$. When the parameter $\beta$ is not close to $\beta_n$, or when $n$ is small, there is a good chance that the solution to the equation $\sum_{j=1}^{n}\psi_j(I(\omega_j),\beta) p_j=0$ doesn't exist which raises some computational issues as mentioned by Chen et al. (2008). To overcome this difficulty, Chen et al. (2008) proposed an adjusted empirical likelihood (AEL) ratio function by adding $\psi_{n+1}$-th term to guarantee the zero to be an interior point of the convex hull so that the required numerical maximization always has a solution. By doing so, they modified Owen's method and applied it to dependent observations with the asymptotic null distribution of the statistic as obtained by Owen. We adopt their idea to modify Yau's method for ARFIMA models.

Denote $\psi_j = \psi_j(\beta) =  \psi_j(I(\omega_j),\beta)$ and $\bar{\psi}_n = \bar{\psi}_n(\beta) = \frac{1}{n}\sum_{j=1}^{n}\psi_j $. For some positive constant $a_n$, define
$$\psi_{n+1} = \psi_{n+1}(\beta) =  -\frac{a_n}{n} \sum_{j=1}^{n}\psi_j = -a_n {\bar{\psi}}_n.$$

Here we choose $a_n=\max(1,\log(n)/2)$ suggested by Chen et al. (2008). Hence the adjusted empirical likelihood ratio for any value $\beta \in \mathbb{B}$ is given by,
$$R(\beta) = \left. {\sup \prod_{j=1}^{n+1} p_j} \middle/ {\prod_{j=1}^{n+1} \frac{1}{n}} \right. = \sup \prod_{j=1}^{n+1} (n+1)p_j,$$
where the maximization is subject to: (i) $\sum_{j=1}^{n+1}\psi_j p_j=0$, (ii) $\sum_{j=1}^{n+1} p_j=1$, and (iii) $p_j \geq 0 $. Similarly, by Lagrange multiplier method we obtain
$$ p_j = [(n+1)\{1 + {\xi}(\beta)^\prime\psi_j\}]^{-1} \qquad j=1,2,...,n+1, $$
where ${\xi}(\beta)$ is the Lagrangian multiplier satisfying,
$$\sum_{j=1}^{n+1}\frac{\psi_j}{1 + {\xi}(\beta)^\prime\psi_j}=0.$$

Thus the adjusted empirical likelihood ratio (AEL) statistic is defined by
\begin{equation}
{W}^*(\beta) = 2 \, \sum_{j=1}^{n+1} \ln \{1 + {\xi}(\beta)^\prime\psi_j\}. \tag{3.1}\label{eq:3.1}
\end{equation}

With argument similar to Yau (2012) and under the regularity conditions (A.1) to (A.6) given by Fox and Taqqu (1986), it can be shown that ${W}^*(\beta)$ has a chi-square distribution which is stated in the following theorem. As mentioned by Yau (2012), the following theorem also applies to ARFIMA process with $d\in [\delta, 0.5 - \delta]$ for any fixed $\delta > 0$, since $\beta$ belongs to a compact space.

\begin{thm}
Let $Z_t$ be an ARFIMA (p,d,q) process defined in \eqref{eq:2.1}, where $\beta \in \mathbb{B} \subseteq \mathbb{R}^k$, with $k=p+q+2$, and satisfying (A.1)-(A.6) given by Fox and Taqqu (1986). Then ${W}^*(\beta) \xrightarrow {d} \chi^2_k$.
\end{thm}
\begin{proof}
The proof is provided in Appendix. 
\end{proof}

\section{Simulations}
\subsection{Confidence Region}
In this section, we compare the confidence regions based on the AEL method and the EL method for ARFIMA models with different sample sizes and different distributions of error terms. 

\subsubsection{ARFIMA (1,d,0) model with different sample sizes}
We consider ARFIMA (1,d,0) model given by:  
$$(1-\phi B) Z_t = (1-B)^{-d} a_t$$
where $a_t$ is the white noise process with mean zero and variance $\sigma^2$. We take $\phi=0.2, d=0.3$. We considered the length $T$ of observations to be 100 and 1500 to illustrate the usefulness of adjusted empirical likelihood under small sample sizes.  ${W}^*(\beta)$ is calculated at different points over the parameter space $(\phi,d) \in \{(0,1) \times(0,0.5)\}$ by taking $a_t \sim N(0,1)$ and 95\% adjusted empirical likelihood confidence regions for the parameters of the model are produced using contour plots based on the critical value of ${\chi}^2_{2,0.95}$.The adjusted empirical likelihood confidence region is compared with the unadjusted empirical likelihood confidence region under each sample size $T$.

Figure 1 shows the 95\% adjusted empirical likelihood (solid line) and unadjusted empirical likelihood (dashed line) confidence regions for the parameters of an ARFIMA(1,d,0) model with $T=100$ and $T=1500$ observations. In each case the mean is subtracted from the white noise processes in order to have mean zero for the error terms. It can be seen that with the same nominal level, the confidence contours for adjusted empirical likelihood contain the ones based on the unadjusted empirical likelihood, especially when the sample size is small. The difference is clearly for small sample size whereas for large sample size the two methods give similar contours but a closer look will ensure that adjusted empirical likelihood gives confidence regions which still contain the unadjusted empirical likelihood confidence region.
\begin{figure}[H]
    \centering
    \begin{subfigure}[b]{0.4\textwidth}
        \includegraphics[width=\textwidth]{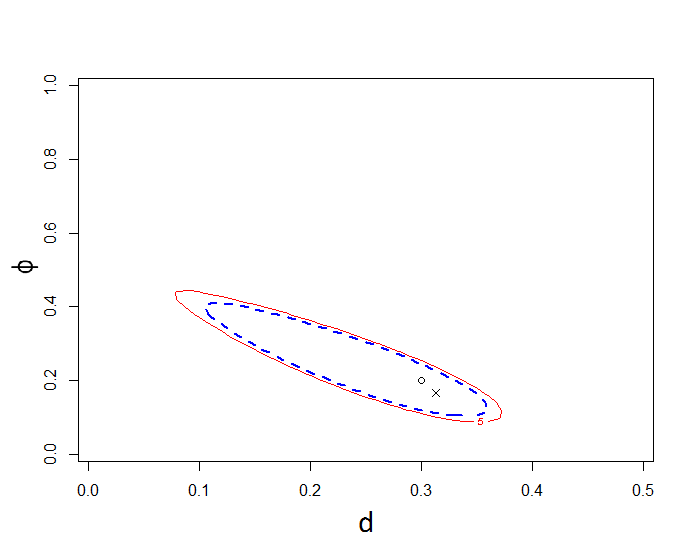}
        \caption{\small T=100}
        \label{fig:T=100}
    \end{subfigure}
    \begin{subfigure}[b]{0.4\textwidth}
        \includegraphics[width=\textwidth]{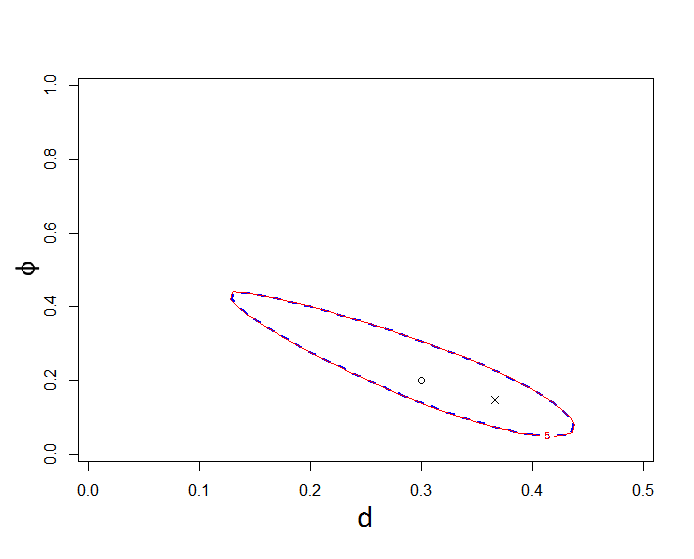}
        \caption{\small T=1500}
        \label{fig:T=1500}
    \end{subfigure}
    \captionsetup{justification=centering}
    \caption{\small 95\% adjusted empirical likelihood (solid line) and unadjusted empirical likelihood (dashed line) confidence regions for the parameters of an ARFIMA(1,d,0) model with $a_t \sim N(0,1)$, `o' is the true parameter, `+' is the estimated value value.}\label{fig1:afrima(1,d,0)}
\end{figure}

\subsubsection{ARFIMA (1,d,0) model with different error distributions}
For ARFIMA(1,d,0) model, we consider the behavior of adjusted empirical likelihood confidence regions under four different distributions for the white noise process $a_t$:~$t_5$,~$t_{10}$,\\~$\exp(1)$ and ${\chi}^2_5$. The latter two distributions are centered around zero. In these cases, we consider $\phi=0.2, d=0.3$ and $T=1500$.

Figure 2 shows the 95\% adjusted empirical likelihood confidence regions for ARFIMA models along with that of the unadjusted empirical likelihood under the above mentioned four different error distributions. In order to have a mean zero for the white noise process, the mean is subtracted in each case. Since the sample size is large, the confidence contours based on the AEL and the EL are most overlapped.  It can be seen that depending on the distribution of the white noise process, the shapes of the contours of the confidence regions are changing in order to adopt the differences in error distributions.

\subsubsection{ARFIMA (0,d,1) model with different error distributions}
We consider ARFIMA (0,d,1) model: 
$$(1-B)^d Z_t = (1+\theta B) a_t,$$
where $a_t$ is the white noise process with mean zero and variance $\sigma^2$. Fix $\phi=0.2, d=0.3$ and $T=1500$. For the adjusted empirical likelihood confidence regions, $\hat{W}^*(\beta)$ are calculated at different points over the parameter space $(\theta,d) \in \{(0,1) \times(0,0.5)\}$ under five different different distributions for the white noise process $a_t$:~N(0,1),~$t_5$,~$t_{10}$,~$\exp(1)$ and ${\chi}^2_5$. The latter two distributions are centered around zero. Figure 3 shows that the shapes of the confidence regions for models under different white noise distributions are different as it changes to adapt the differences in error distributions depicting the non-parametric property of adjusted empirical likelihood.

\subsection{Coverage Probabilities}
In this section, a Monte Carlo experiment is conducted to explore the accuracy of the adjusted empirical likelihood confidence regions for ARFIMA model in terms of coverage probability. To make a fair comparison to the EL method proposed by Yau (2012) and the other unadjusted EL method with and without Bartlett correction, 
we consider ARFIMA(0,d,0) model: $ Z_t = (1-B)^{-d} a_t$ with various sample sizes, values of~$d$~and distributions of the white noise term. The simulations are carried out under two different distributions for the error terms, $a_t$: $N(0,1)$ and $t_5$. In both cases the mean is subtracted from the white noise process in order to make the white noise process to have mean zero. The simulations are conducted for different values of $d=(0.1, 0.2, 0.3, 0.4, 0.49)$. Under each case, 1000 series of size $T=(50, 70, 100, 200)$ are drawn and the coverage probabilities are computed. We choose $a_n=\log(n)/2$ as in the definition of $\psi_{n+1}$.

The adjusted empirical likelihood coverage probabilities are compared with the unadjusted empirical likelihood coverage probabilities. The coverage probabilities of intervals based on theoretical and estimated Bartlett-correction (DiCiccio et al., 1991) are also computed for comparison purpose. Table 1 provides the results for the nominal level of 95\%. It can be seen that the coverage probabilities of the adjusted empirical likelihood are closer to the nominal value of 0.95 under each sample size and error distribution considered. For small sample sizes, the adjusted empirical likelihood gives more accurate results than the unadjusted empirical likelihood method. Further it shows that although the theoretical and estimated Bartlett-correction methods give improved results than that of the unadjusted empirical likelihood method, neither seems to give better results than adjusted empirical likelihood method.

\begin{figure}[H]
    \centering
    \begin{subfigure}[b]{0.3\textwidth}
        \includegraphics[width=\textwidth]{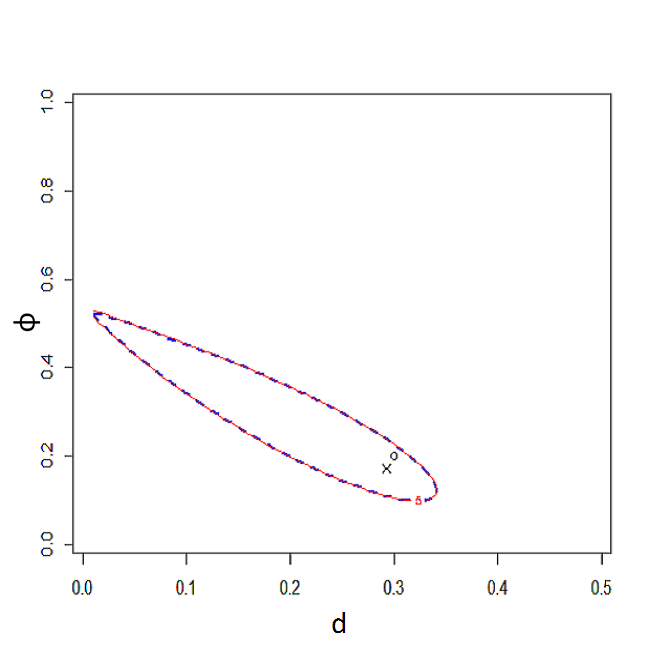}
        \caption{\small $a_t \sim t_5$}
        \label{fig:t5}
    \end{subfigure}
    \begin{subfigure}[b]{0.3\textwidth}
        \includegraphics[width=\textwidth]{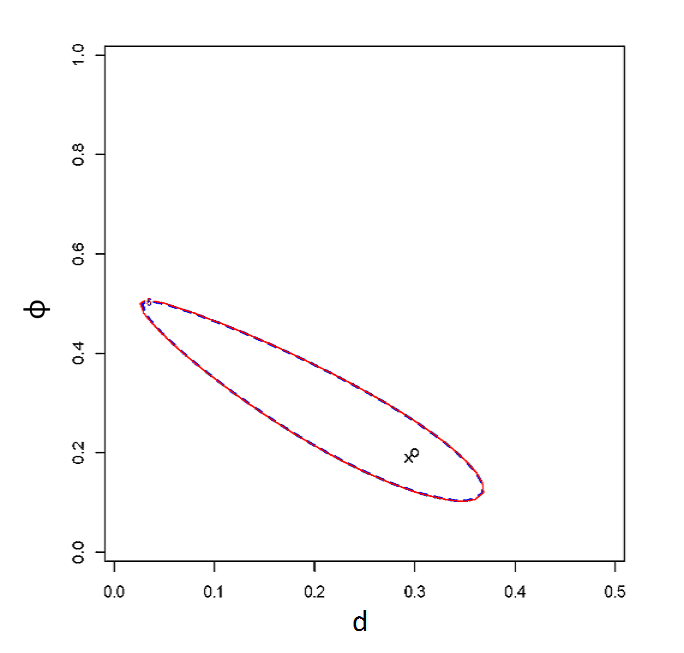}
        \caption{\small $a_t \sim t_{10}$}
        \label{fig:t10}
    \end{subfigure}
    \begin{subfigure}[b]{0.3\textwidth}
        \includegraphics[width=\textwidth]{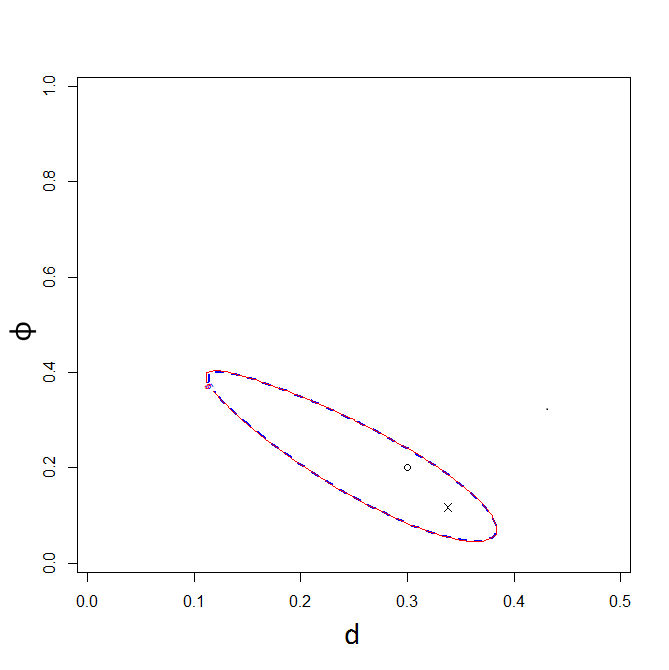}
        \caption{\small $a_t \sim \chi^2_5$}
        \label{fig:exp(1)}
    \end{subfigure}
        \begin{subfigure}[b]{0.3\textwidth}
        \includegraphics[width=\textwidth]{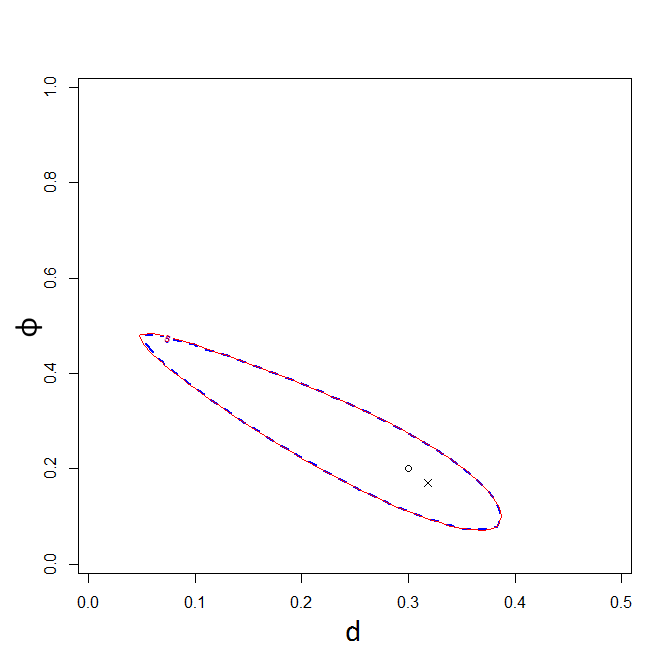}
        \caption{\small $a_t \sim \exp(1)$}
        \label{fig:chisq(5)}
    \end{subfigure}
    \captionsetup{justification=centering}
    \caption{\small 95\%adjusted empirical likelihood (solid line) and unadjusted empirical likelihood (dashed line) confidence regions for the parameters of an ARFIMA(1,d,0) model, `o' is the true parameter, `+' is the estimated value value.}\label{fig2:afrima(1,d,0)}
\end{figure}

\section{Discussion}
In this paper, we propose an adjusted empirical likelihood to extend Yau's (2012) method for long-memory time series models, specifically ARFIMA models, by adopting the idea of Chen et al. (2008). The asymptotic null distribution of the adjusted empirical likelihood statistic for long-memory time series models has been established as a standard chi-square distribution. Confidence contours for ARFIMA(1,d,0) and ARFIMA(0,d,1) models based on AEL and EL methods are drawn with different sample sizes and different error distributions to illustrate the comparison. Simulations for ARFIMA(0,d,0) with different distributions for white noise process have been carried out to illustrate the performance of the proposed AEL method. Coverage probabilities of the AEL method have been compared to the unadjusted EL method with and without estimated and theoretical Bartlett-corrected ones under different sample sizes. The results indicate that the proposed AEL method compares favorably with other methods, especially when the sample size is small.

\begin{figure}[H]
    \centering
    \begin{subfigure}[b]{0.3\textwidth}
        \includegraphics[width=\textwidth]{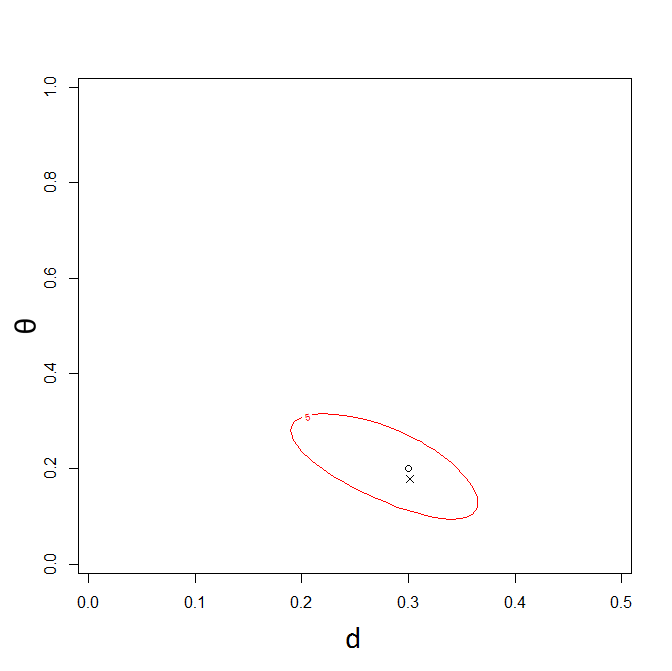}
        \caption{\small $a_t \sim N(0,1)$}
        \label{fig:t5}
    \end{subfigure}
    \begin{subfigure}[b]{0.3\textwidth}
        \includegraphics[width=\textwidth]{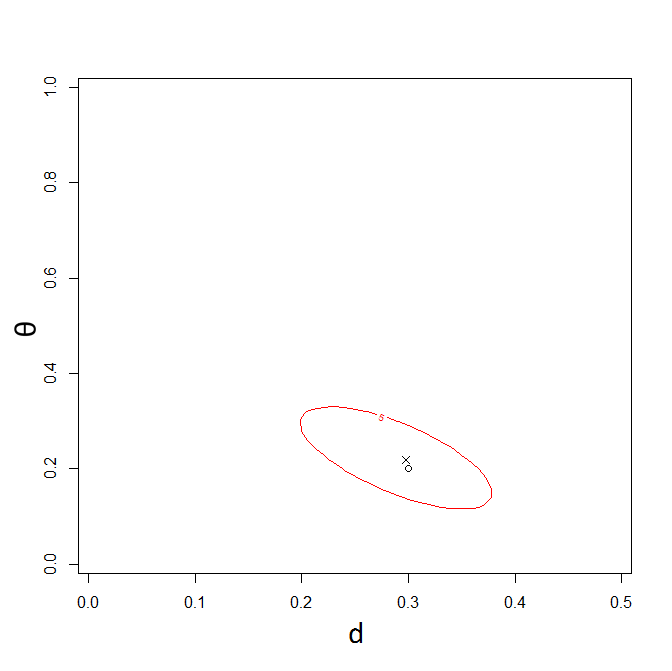}
        \caption{\small $a_t \sim t_5$}
        \label{fig:t5}
    \end{subfigure}
    \begin{subfigure}[b]{0.3\textwidth}
        \includegraphics[width=\textwidth]{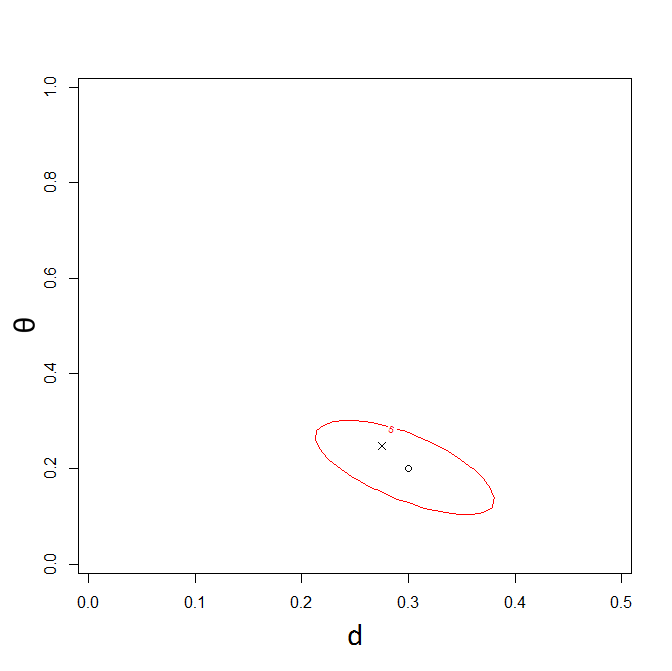}
        \caption{\small $a_t \sim t_{10}$}
        \label{fig:t10}
    \end{subfigure}
    \begin{subfigure}[b]{0.3\textwidth}
        \includegraphics[width=\textwidth]{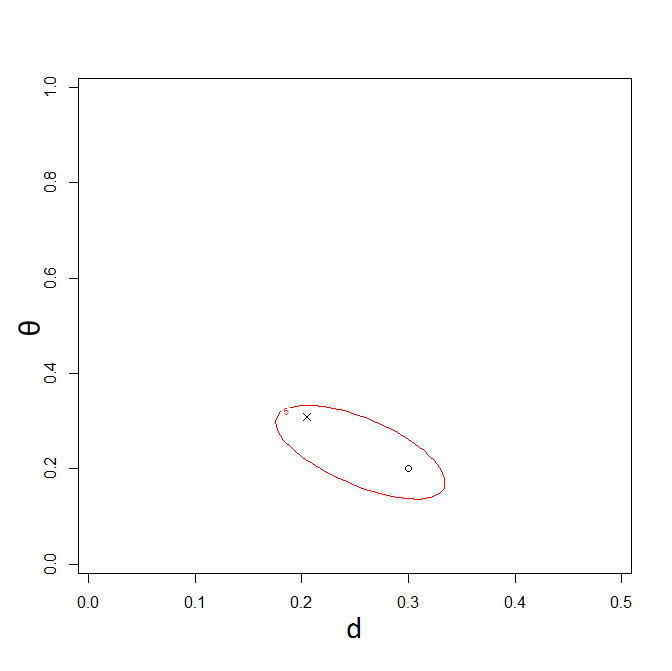}
        \caption{\small $a_t \sim \chi^2_5$}
        \label{fig:exp(1)}
    \end{subfigure}
        \begin{subfigure}[b]{0.3\textwidth}
        \includegraphics[width=\textwidth]{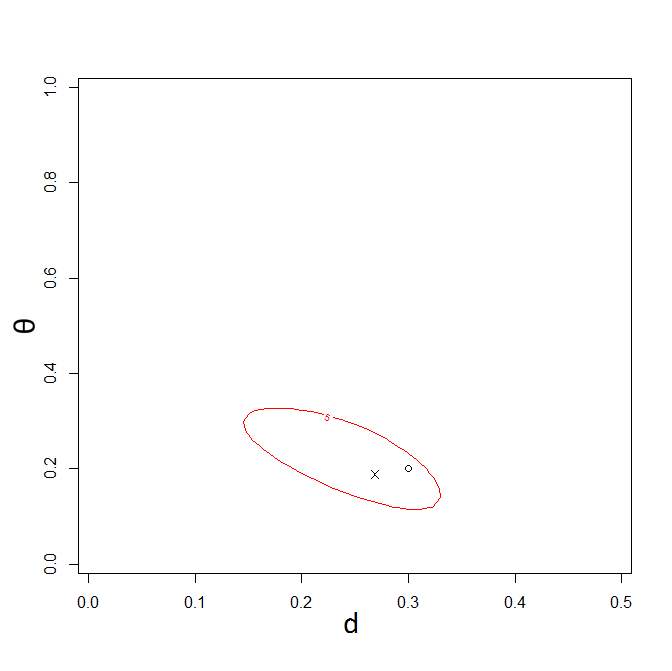}
        \caption{\small $a_t \sim \exp(1)$}
        \label{fig:chisq(5)}
    \end{subfigure}
    \captionsetup{justification=centering}
    \caption{\small 95\%adjusted empirical likelihood confidence regions for the parameters of an ARFIMA(0,d,1) model, `o' is the true parameter, `+' is the estimated value value.}\label{fig:afrima(0,d,1)}
\end{figure}

\begin{table}[H]
\caption{\textit{\textsc{arfima} (0,d,0) model} - with n=(T-1)/2} 
\centering
\small
\begin{tabular}{c c c c c c c}
\hline\hline 
T & Method & d=0.1 & d=0.2 & d=0.3 & d=0.4 & d=0.49 \\ [0.5ex]
\hline
\multicolumn{7}{ c }{Model: $a_t\sim N(0,1)$} \\ [0.5ex]
\multirow{4}{*}{T=50} & EL & 0.847 & 0.838 & 0.834 & 0.838 & 0.840 \\
 & EB & 0.861 & 0.854 & 0.843 & 0.849 & 0.853 \\
 & TB & 0.871 & 0.861 & 0.857 & 0.854 & 0.859 \\
 & AEL & 0.884 & 0.876 & 0.874 & 0.871 & 0.875 \\ [0.5ex]
 \multirow{4}{*}{T=70} & EL & 0.863 & 0.860 & 0.847 & 0.852 & 0.867 \\
 & EB & 0.871 & 0.866 & 0.859 & 0.860 & 0.878 \\
 & TB & 0.879 & 0.872 & 0.863 & 0.863 & 0.883 \\
 & AEL & 0.887 & 0.878 & 0.876 & 0.877 & 0.894 \\ [0.5ex]
\multirow{4}{*}{T=100}& EL & 0.890 & 0.888 & 0.892 & 0.891 & 0.894 \\
 & EB & 0.896 & 0.893 & 0.896 & 0.898 & 0.899 \\
 & TB & 0.900 & 0.897 & 0.898 & 0.900 & 0.902 \\
 & AEL & 0.899 & 0.902 & 0.904 & 0.903 & 0.902 \\ [0.5ex]
\multirow{4}{*}{T=200}& EL & 0.905 & 0.907 & 0.908 & 0.907 & 0.911 \\
 & EB & 0.910 & 0.915 & 0.912 & 0.915 & 0.917 \\
 & TB & 0.913 & 0.918 & 0.916 & 0.925 & 0.926 \\
 & AEL & 0.920 & 0.916 & 0.914 & 0.916 & 0.917 \\ [1ex]

\multicolumn{7}{ c }{Model: $a_t\sim t_5$} \\ [0.5ex]
\multirow{4}{*}{T=50} & EL & 0.871 & 0.871 & 0.865 & 0.860 & 0.861 \\
 & EB & 0.880 & 0.878 & 0.869 & 0.869 & 0.873 \\
 & TB & 0.889 & 0.884 & 0.877 & 0.871 & 0.882 \\
 & AEL & 0.898 & 0.894 & 0.891 & 0.885 & 0.888 \\ [0.5ex]
 \multirow{4}{*}{T=70} & EL & 0.869 & 0.867 & 0.872 & 0.876 & 0.887 \\
 & EB & 0.879 & 0.880 & 0.79 & 0.881 & 0.892 \\
 & TB & 0.886 & 0.883 & 0.881 & 0.886 & 0.896 \\
 & AEL & 0.893 & 0.887 & 0.887 & 0.896 & 0.904 \\ [0.5ex]
\multirow{4}{*}{T=100} & EL & 0.897 & 0.894 & 0.898 & 0.896 & 0.899 \\
 & EB & 0.905 & 0.897 & 0.902 & 0.903 & 0.904 \\
 & TB & 0.909 & 0.901 & 0.905 & 0.904 & 0.906 \\
 & AEL & 0.908 & 0.909 & 0.907 & 0.908 & 0.906 \\ [0.5ex]
\multirow{4}{*}{T=200} & EL & 0.912 & 0.912 & 0.913 & 0.915 & 0.910 \\
 & EB & 0.913 & 0.913 & 0.917 & 0.918 & 0.918 \\
 & TB & 0.917 & 0.914 & 0.918 & 0.925 & 0.925 \\
 & AEL & 0.918 & 0.914 & 0.919 & 0.919 & 0.917 \\ [1ex] \hline
\multicolumn{7}{ l }{EL= empirical likelihood; TB=EL with theoretical Bartlett correction;} \\
\multicolumn{7}{ l }{EB=EL with estimated Bartlett correction; AEL=Adjusted EL.} \\
\end{tabular}
\label{table:Table1}
\end{table}

\section*{Appendix}
In this section, we give the brief proof of Theorem 1. Detailed proof is available upon request. 
\begin{proof}
\noindent First we prove that $\xi=O_p(n^{-\frac{1}{2}})$. The adjusted empirical likelihood ratio function is,
$$W^*(\theta) = -2\,\sup \bigg \{ \sum_{j=1}^{n+1}log[(n+1)p_j]| p_j \geq 0,j=1,...n+1;\sum_{j=1}^{n+1}p_j=1; \sum_{j=1}^{n+1}\psi _j(I(\omega _j),\beta )=0 \bigg\} $$
where $\psi_{n+1} = -\frac{a_n}{n} \sum_{j=1}^{n}\psi_j = -a_n {\bar{\psi}}_n$ and $a_n=\max(1,log(n)/2)=o_p(n)$. We will show that $W^*(\theta) \sim {\chi_k}^2$. First we need to show that $\xi=O_p(n^{-\frac{1}{2}})$. Denote $\psi_j \equiv  \psi_j(I(\omega_j),\beta)$. Assume $Var\{\psi(I(\omega),\beta)\}$ is finite and has rank $q<m (=dim(\psi))$. Let the eigenvalues of $Var\{\psi(I(\omega),\beta)\}$ be ${\sigma_1}^2\leq {\sigma_2}^2 \leq ...\leq {\sigma_m}^2$.
WLOG, assume ${\sigma_1}^2 = 1 $. Let $\xi$ be the solution of
\begin{equation}
\sum_{j=1}^{n+1} \frac{\psi_j}{1+\xi^\prime \psi_j} = 0. \tag{A.1}\label{eq:A1}
\end{equation}
Let $\psi^* = \smash{\displaystyle\max_{1 \leq j \leq n}} \parallel \psi_j\parallel$. Since $\left| \psi_j \right| \geq0, j=1,...,n$ are independent, by Lemma 3 of Owen (1990), we have $\psi^*= \smash{\displaystyle\max_{1 \leq j \leq n}} \parallel \psi_j\parallel = o_p(n^{\frac{1}{2}})$ if $E({\left|\psi_j\right|}^2)<\infty.$ By CLT, we have ${\bar{\psi}}_n=\frac{1}{n}\sum_{j=1}^{n}\psi_j=O_p(n^{-\frac{1}{2}})$. 

Let $\xi=\rho\theta$ where $\rho\geq0$ and $\parallel\theta\parallel=1$. Multiplying both sides of \eqref{eq:A1} by $n^{-1}\theta^\prime$, we obtain 
\begin{align*}
0&= n^{-1}\theta^\prime \sum_{j=1}^{n+1}\frac{\psi_j}{1+\xi^\prime\psi_j} = \frac{\theta^\prime}{n}\sum_{j=1}^{n+1}\bigg[ \psi_j - \frac{\xi^\prime{\psi_j}^2}{1+\xi^\prime\psi_j}\bigg]\\
&= \frac{\theta^\prime}{n}\sum_{j=1}^{n+1}\psi_j - \frac{\rho}{n} \sum_{j=1}^{n+1}\frac{({\theta^\prime\psi_j})^2}{1+\rho\theta^\prime\psi_j}\\
&\leq \theta^\prime \bar{\psi}_n \bigg( 1-\frac{a_n}{n}\bigg) - \frac{\rho}{n(1+\rho\psi^*)} \sum_{j=1}^{n}(\theta^\prime\psi_j)^2 \tag{$(n+1)^{th}$ term in the second summation is non-negative}\\
&= \theta^\prime \bar{\psi}_n - \theta^\prime {\bar{\psi}}_n \frac{a_n}{n} - \frac{\rho}{n(1+\rho\psi^*)}\sum_{j=1}^{n}(\theta^\prime\psi_j)^2\\
&=\theta^\prime \bar{\psi}_n - \frac{\rho}{n(1+\rho\psi^*)}\sum_{j=1}^{n}(\theta^\prime\psi_j)^2 + O_p(n^{-\frac{3}{2}}a_n). \tag{A.2} \label{eq:A2}
\end{align*}

\noindent The assumption on $Var\{\psi(I(\omega),\beta)\}$ implies that
\begin{align*}
\frac{1}{n}\sum_{j=1}^{n}(\theta^\prime\psi_j)^2 \geq (1-\epsilon){\sigma_1}^2 = 1-\epsilon
\end{align*}
in probability for some $0 < \epsilon <1$. So as long as $a_n = o_p(n)$, \eqref{eq:A2} implies that
\begin{align*}
\frac{\rho}{1+\rho\psi^*} &\leq \theta^\prime\bar{\psi}_n(1-\epsilon)^{-1} = O_p(n^{-\frac{1}{2}}).
\end{align*}
Since
$$\theta^\prime \bar{\psi}_n - \frac{\rho}{n(1+\rho\psi^*)}\sum_{j=1}^{n}(\theta^\prime\psi_j)^2 \geq 0, $$
$$\theta^\prime\bar{\psi}_n \geq \frac{\rho}{(1+\rho\psi^*)}\frac{\sum_{j=1}^{n}(\theta^\prime\psi_j)^2}{n} \geq \frac{\rho}{(1+\rho\psi^*)} (1-\epsilon). $$
\noindent Hence $$\frac{\rho}{(1+\rho\psi^*)} \leq \theta^\prime \bar{\psi}_n (1-\epsilon)^{-1}.$$
Therefore, $\rho = \parallel\xi\parallel=O_p(n^{-\frac{1}{2}})$ which implies $\xi=O_p(n^{-\frac{1}{2}})$.\\

Now we need to prove that $W^*(\theta) \sim {\chi_k}^2$. Under suitable regularity conditions (Dzhaparidze, 1986), $n^{-\frac{1}{2}} (\beta_n - \beta) \sim N(0,V)$ asymptotically, where $V$ is the covariance matrix of $\beta$ and $\beta_n$ is the Whittle's estimator of $\beta$. Therefore, $\beta - \beta_n = O_p(n^{-\frac{1}{2}})$ which implies $\beta = \beta_n + n^{-\frac{1}{2}}u$ with $\left| u \right| < +\infty$. From \eqref{eq:A1}, we have
\begin{align*}
0 &=\sum_{j=1}^{n+1} \frac{\psi_j}{1+\xi^\prime \psi_j}\\
&= \frac{1}{n}\sum_{j=1}^{n+1}\psi_j [1 - \xi^\prime \psi_j]\\
&= \bar{\psi}_n + \frac{1}{n} \psi_{n+1} - \frac{1}{n}\sum_{j=1}^{n+1}\psi_j \xi^\prime {\psi_j}^\prime \tag{$\frac{1}{n} \psi_{n+1} = o_p(n^{-\frac{1}{2}})$} \\
&= \bar{\psi}_n - \frac{1}{n}\sum_{j=1}^{n}\psi_j \xi^\prime {\psi_j}^\prime + o_p(n^{-\frac{1}{2}}). \tag{A.3} \label{eq:A3}
\end{align*}
Using Taylor expansion at $\beta=\beta_n$, we have,
\begin{equation}
\bar{\psi}_n = \frac{1}{n}\sum_{j=1}^{n}\psi_j = \frac{1}{n}\sum_{j=1}^{n}\frac{\partial \psi_j(I(\omega_j),t)}{\partial t^\prime} \bigg|_{\scriptscriptstyle \beta_n} (\beta-\beta_n) + O_p(n^{-1}) = \hat{A}(\beta_n) (\beta-\beta_n) + o_p(n^{-\frac{1}{2}}), \tag{A.4} \label{eq:A4}
\end{equation}
where 
$$\hat{A}(\beta_n) = \frac{1}{n}\sum_{j=1}^{n}\frac{\partial \psi_j(I(\omega_j),t)}{\partial t^\prime} \bigg|_{\scriptscriptstyle \beta_n}.$$
Since $\psi_j(I(\omega_j),\beta_n) \rightarrow \psi_j(I(\omega_j),\beta)$ in probability, from Monti (1997) we have,
\begin{equation}
\frac{1}{n}\sum_{j=1}^{n}\psi_j {\psi_j}^\prime = \frac{1}{n}\sum_{j=1}^{n}\psi_j(I(\omega_j),\beta_n){\psi_j(I(\omega_j),\beta_n)}^\prime + O_p(n^{-\frac{1}{2}}) = \hat{\Sigma}(\beta_n) + O_p(n^{-\frac{1}{2}}), \tag{A.5} \label{eq:A5}
\end{equation}
where 
$$\hat{\Sigma}(\beta_n) = \frac{1}{n}\sum_{j=1}^{n}\psi_j(I(\omega_j),\beta_n){\psi_j(I(\omega_j),\beta_n)}^\prime.$$
\noindent With \eqref{eq:A3}, \eqref{eq:A4} and \eqref{eq:A5} we have,
$$\hat{A}(\beta_n) (\beta-\beta_n) + o_p(n^{-\frac{1}{2}}) - \xi^\prime \hat{\Sigma}(\beta_n) - o_p(1) + o_p(n^{-\frac{1}{2}}) = 0,$$
$$ \Rightarrow \xi^\prime \hat{\Sigma}(\beta_n) = \hat{A}(\beta_n) (\beta-\beta_n) +  o_p(n^{-\frac{1}{2}}),$$
\begin{equation}
\Rightarrow \xi = { \hat{\Sigma}(\beta_n)}^{-1} \hat{A}(\beta_n) (\beta-\beta_n) +  o_p(n^{-\frac{1}{2}}). \tag{A.6} \label{eq:A6}
\end{equation}
By Taylor expansion, we obtain,
\begin{align*}
W^*(\beta) &= 2 \sum_{j=1}^{n+1}\log [1+\xi^\prime \psi_j(I(\omega_j),\beta)]\\
&= 2\sum_{j=1}^{n+1} \big[ \xi^\prime \psi_j(I(\omega_j),\beta) - \frac{1}{2} {(\xi^\prime \psi_j(I(\omega_j),\beta))}^2 \big] + o_p(1).
\end{align*}
and using \eqref{eq:A4}, \eqref{eq:A5} and \eqref{eq:A6} gives,
\begin{align*}
W^*(\beta) = n(\beta-\beta_n)^\prime \hat{V}^{-1} (\beta-\beta_n) + o_p(1),  \tag{A.7} \label{eq:A7}
\end{align*}
where
$$ \hat{V} = {\hat{A}(\beta_n)}^{-1} \hat{\Sigma}(\beta_n) {\{ {\hat{A}(\beta_n)}^\prime\} }^{-1}, $$
is a consistent estimator of the covariance matrix $V$ of $\beta$ which is proved in Appendix 1 in Monti (1997). Hence $W^*(\beta)$ converges to a standard chi-square distribution with $k$ degrees of freedom as $n \rightarrow \infty$.\\

\end{proof}

\end{document}